\documentclass[journal]{IEEEtran}
\usepackage{amsmath}
\usepackage{amssymb}
\usepackage{amsthm}
\usepackage{graphicx}
\usepackage{epstopdf}
\usepackage{cite}
\usepackage{algorithm}
\usepackage{algpseudocode}
\usepackage{subfigure}
\usepackage{enumerate}

\theoremstyle{plain}
\newtheorem{theorem}{Theorem}

\newtheorem{lemma}[theorem]{Lemma}
\newtheorem{proposition}[theorem]{Proposition}

\theoremstyle{definition}

\theoremstyle{remark}
\newtheorem{remark}[theorem]{Remark}

\begin{document}

\title{Reconstruction of Correlated Sources with Energy Harvesting Constraints in Delay-constrained and Delay-tolerant Communication Scenarios}
\author{Miguel Calvo-Fullana, Javier Matamoros, and Carles Ant\'on-Haro
\thanks{
This work was partly sponsored by the Catalan Government under grant SGR2014-1567, the Spanish Government under projects PCIN-2013-027 (E-CROPS) and TEC2013-44591-P (INTENSYV), and the European Commission under Grant Agreement 318306 (NEWCOM\#).

The authors are with the Centre Tecnol\`ogic de Telecomunicacions de Catalunya (CTTC/CERCA), 08860 Castelldefels, Barcelona, Spain (e-mail: \{miguel.calvo, javier.matamoros, carles.anton\}@cttc.cat).
}
}

\maketitle

\begin{abstract}
In this paper, we investigate the reconstruction of time-correlated sources in a point-to-point communications scenario comprising an energy-harvesting sensor and a Fusion Center (FC). Our goal is to minimize the average distortion in the reconstructed observations by using data from previously encoded sources as side information. First, we analyze a delay-constrained scenario, where the sources must be reconstructed before the next time slot. We formulate the problem in a convex optimization framework and derive the optimal transmission (i.e., power and rate allocation) policy. To solve this problem, we propose an iterative algorithm based on the subgradient method. Interestingly, the solution to the problem consists of a coupling between a two-dimensional directional water-filling algorithm (for power allocation) and a reverse water-filling algorithm (for rate allocation). Then we find a more general solution to this problem in a delay-tolerant scenario where the time horizon for source reconstruction is extended to multiple time slots. Finally, we provide some numerical results that illustrate the impact of delay and correlation in the power and rate allocation policies, and in the resulting reconstruction distortion. We also discuss the performance gap exhibited by a heuristic \emph{online} policy derived from the optimal (offline) one.
\end{abstract}

\IEEEpeerreviewmaketitle

\section{Introduction}
\label{sec:Introduction}

Sensor nodes are usually powered by batteries which can be costly, difficult or even impossible to replace (e.g., when nodes are deployed in remote locations). In recent years, energy harvesting has emerged as a technology capable of overcoming (or, at least, alleviating) the limitations imposed by non-rechargeable batteries. Specifically, nodes equipped with an energy harvesting device are capable of scavenging e.g., solar, wind, thermal, kinetic energy from the environment \cite{vullers2010energy} and, by doing so, extend their operational lifetime.

Energy harvesting has received considerable attention by the wireless communications and information theory communities (see \cite{ulukus2015energy} and references therein for an overview of current advances). For point-to-point scenarios, and under the assumption of known energy and data arrivals (\textit{offline} optimization), the main focus has been on the derivation of optimal transmission strategies at the sensor node. In \cite{yang2012optimal}, the authors study the problem of minimizing the time by which all data packets are transmitted to the destination. A number of authors go one step beyond and investigate the impact of \emph{finite} energy storage capacity \cite{tutuncuoglu2012optimum} or battery leakage \cite{devillers2012general}; generalize the analysis to fading channels \cite{ozel2011transmission}; or explicitly take into consideration the energy needed for data processing (in addition to data transmission) \cite{orhan2013optimal}. Further, other communication scenarios have been investigated such as the broadcast \cite{ozel2012optimal} or the multiple access \cite{yang2012optimalMAC} channels.

For uncoded transmissions, \cite{cui2007estimation} investigates a number of energy-related aspects in a context of wireless sensor networks for parameter estimation. As for the coded case, in \cite{draper2004side} the authors generalize Wyner-Ziv's source coding strategies with side information \cite{wyner1976rate} to tree-structured sensor networks.

Several aspects of source and channel coding have been analyzed in \emph{energy harvesting} scenarios. A point-to-point case was studied in \cite{castiglione2012energy}, where rate-distortion allocation is optimized for stationary energy arrivals under data queue stability. These results were extended in \cite{castiglione2014energy} to the case of \emph{finite} energy and data buffers. Besides, the multi-hop scenario was studied in \cite{tapparello2012dynamic} with \emph{correlated} sources and \emph{distributed} source coding. From a finite-horizon point of view, in \cite{orhan2013delay} the problem of minimizing the reconstruction distortion of a Gaussian source is considered.

\subsection{Contribution}

In this paper, we investigate the reconstruction of time-correlated
sources in a point-to-point communications scenario\footnote{Other communication scenarios such as multiple-access (MAC) channels, which are indeed relevant for wireless sensor networks, are left for future work in this area.}. As in \cite{yang2012optimal,orhan2013delay,ozel2012optimal,devillers2012general} we assume that energy arrivals are non-causally known, thus taking an \emph{offline} optimization approach to the problem. Hence, the
solution turns out to be a benchmark against which \emph{any} online policy can be compared (we also propose one heuristic online policy here). Overall, the main contributions of this work are as follows
\begin{itemize}
\item We consider \emph{time-correlated} sources. The introduction of temporal correlation in the sources is particularly relevant for video coding applications \cite{puri2006distributed} since, in this case, images (i.e., sources) in consecutive frames are clearly correlated. Video source coding has been widely investigated in the literature \cite{ishwar2004decoder,ma2011delayed,yang2011rate}. In \cite{ishwar2004decoder} the authors model video signals as a sequence of time-correlated (correlation given by a first-order auto-regressive process) spatially independent and identically distributed Gaussian processes (namely, frames). Such correlation model, which we adopt in this work, is illustrated in Fig. \ref{fig:systemModelExample}. Still, more general correlation models are also available \cite{yang2011rate}.
\item We consider \emph{delay-tolerant} reconstruction scenarios. In other works \cite{ma2011delayed}, the authors analyze the impact of a \emph{delay-tolerant} reconstruction of the correlated source. Those studies, however, were conducted in scenarios without energy harvesting. Our work goes one step beyond and incorporates energy harvesting constraints (in the sensor node) in the source coding process itself. Consequently, the closed-form expressions that we derive for the power and rate allocation policies explicitly take correlation into account. In this respect, we generalize the results of \cite{orhan2013delay} to the correlated case.
\item In contrast to previous works, we further leverage on \emph{side information}-aware coding strategies for WSNs \cite{draper2004side} to exploit correlation. First, we study the \emph{delay-constrained} case \cite{calvo2015reconstruction} in which the source must be reconstructed at the FC before the next time slot. Then we generalize our study to the delay-tolerant case\textemdash where the time horizon for source reconstruction is extended to multiple time slots. For both cases, we derive the optimal transmission policy which minimizes the average reconstruction distortion at the destination. Our policy reverts to that of \cite{orhan2013delay} for uncorrelated sources, and to that of \cite{yang2012optimal} for the uncorrelated and delay-constrained case. In order to compute this rate and power allocation policy, we propose an iterative algorithm based on the subgradient method \cite{bertsekas1999nonlinear}. Interestingly, we show that the procedure encompasses the interaction between a two-directional directional water-filing and a reverse water-filling \cite[Chapter 10]{cover2012elements} schemes.
\end{itemize}

Besides, we provide extensive numerical results which illustrate the impact of correlation and delay in the transmit policy and the resulting reconstruction distortion. 

\begin{figure}[t]
    \centering
    \includegraphics[width=\columnwidth]{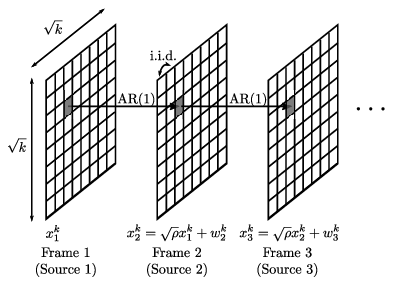} 
    \caption{Temporal and spatial correlation models in video coding.}
    \label{fig:systemModelExample}
\end{figure}

The remainder of this paper is organized as follows. In Section \ref{sec:SystemModel} we introduce the system model and provide details on the encoding process. In Section \ref{sec:MinDistd1}, we address the distortion minimization problem for the delay-constrained case. We formulate the problem as a convex program and derive the optimal power and rate allocation policy. In order to compute this resulting transmission policy, we propose in Section \ref{sec:MinDistd1Alg} an iterative algorithm based on the subgradient method. Next, in Section \ref{sec:MinDist} we generalize the problem (and the solution) to the delay-tolerant case. We provide numerical results in Section \ref{sec:NumRes}, where the effect of correlation as well as delay on the transmit policy and the resulting distortion are assessed. Finally, Section \ref{sec:Conclusions} closes the paper by providing some concluding remarks.

\section{System Model}
\label{sec:SystemModel}
\begin{figure}[t]
    \centering
    \includegraphics[width=\columnwidth]{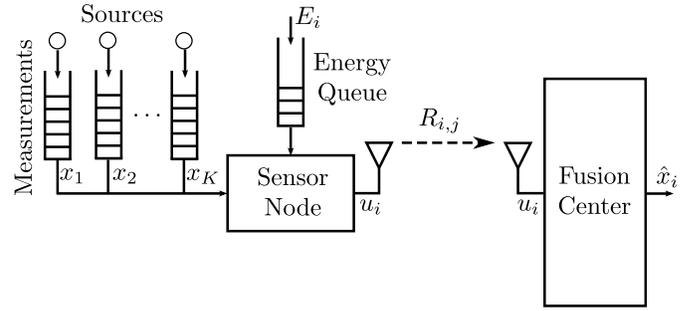}
    \caption{System Model.} 
    \label{fig:systemModel}
\end{figure}

Consider the point-to-point communications scenario depicted in Fig. \ref{fig:systemModel} which comprises one energy-harvesting (EH) sensor and one Fusion Center (FC). We adopt a slotted transmission model, with $K$ denoting the total number of time slots. The sensor measures a time-varying phenomenon of interest which, in the sequel, we model by multiple correlated and memoryless Gaussian wide-sense stationary sources (see rationale in the preceding section). Specifically, each source models the phenomenon in a given time slot. In the $k$-th time slot, the sensor node (i) collects a \emph{large} number of independent and identically distributed (i.i.d.) samples from the $k$-th source; and (ii) encodes those measurements. The encoded data is then transmitted to the FC in $d$ consecutive time slots.

In this work, we consider both delay-constrained ($d=1$) and delay-tolerant ($d>1$) communication scenarios. Clearly, in delay-tolerant scenarios the encoded data transmitted in a given time slot corresponds to multiple sources, as Figure \ref{fig:systemModelTX} illustrates. Let $R_{i,j}$ denote the average transmission rate assigned to the encoded samples of the $j$-th source in the $i$-th time slot. Necessarily, the sum-rate in the $i$-th time slot is upper bounded by the channel capacity\footnote{For the ease of notation, we let the number of channel uses to be equal to the number of samples collected in a given time slot. This number, in turn, is assumed to be large enough to satisfy Shannon's source coding theorem.}, namely,
\begin{align}
\sum\limits_{j=i-d+1}^{i}R_{i,j} \leq \log\left(1+|h_i|^2
p_i\right), \quad\quad i=1,\dots, K, \label{eq:CapacityConstraint}
\end{align}
with $|h_i|^2$ and $p_i$ standing for the channel gain and average transmit power in time slot $i$, respectively (channel noise is assumed to be Gaussian-distributed, with zero-mean and unit variance). The $n$ i.i.d. samples collected by the sensor node from the $i$-th source will be denoted in the sequel by $\{x_i^k\}_{k=1}^{n}$. Such samples, we assume, are correlated over time slots through a first-order autoregressive process. Hence, for the $k$-th sample from the $i$-th source we have that
\begin{align}
x^{k}_{i}&=\sqrt{\rho} x^{k}_{i-1}+w^{k}_{i},&       k=1,\dots,n,
\atop i=1,\dots,K,
\end{align}
with $\rho= \mathbb{E}\left[x^{k}_{i}x^{k}_{i-1}\right]$ denoting the correlation coefficient, and $w^{k}_{i}$ standing for an i.i.d. zero-mean Gaussian random variable with variance $\sigma_w^2=(1-\rho)\sigma_x^2$.

\begin{figure}[t]
    \centering
    \includegraphics[width=\columnwidth]{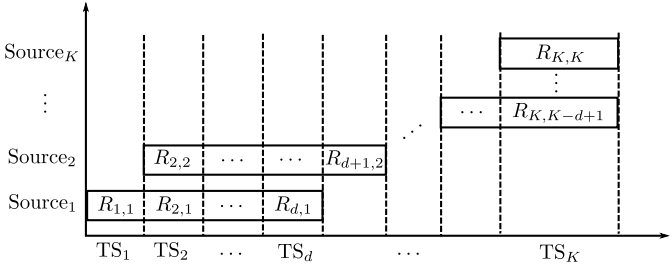}
    \caption{Simultaneous transmission of source measurements.}
    \label{fig:systemModelTX}
\end{figure}

As for the underlying energy harvesting process, we model it as a counting process \cite{yang2012optimal,tutuncuoglu2012optimum} with packet energy arrivals of $E_{i}$ Joules at the beginning of time slot $i$. For simplicity, we assume that energy can be stored in a rechargeable battery of infinite capacity. By considering transmit power as the only energy cost, any transmission (power allocation) policy $\{p_{i}\}$  at the sensor node must satisfy the following energy causality constraint:
\begin{align}
T_{s}\sum_{j=1}^{i}p_{j}\leq\sum_{j=1}^{i}E_{j},        \quad\quad\quad i=1,\dots,K,
\label{eq:ECC}
\end{align}
where $T_{s}$ denotes the duration of the time slot which, in the sequel, we normalize (i.e., $T_{s}=1$).
\begin{remark}
There exist more sophisticated power consumption models encompassing non-ideal circuit power consumption effects \cite{xu2014throughput}  or the impact of processing power \cite{orhan2013optimal} (See also references \cite{cui2005energy} and \cite{wei2015full} for more insight into these matters). In this work, for simplicity, we restrict ourselves to \emph{transmit} power consumption. Nonetheless, the proposed model easily adapts to a constant circuit power consumption. Any model of the form $T_s\sum_{j=1}^i p_j+$ $T_s\sum_{j=1}^i  P^c_i \leq \sum_{j=1}^i E_j$, with  $P^c_i$ being the circuit power consumption at the $i$-th time slot, can be rewritten as $T_s\sum_{j=1}^i p_j \leq \sum_{j=1}^i \bar{E}_j$, where we have defined a new energy harvesting process as $\bar{E}_i=E_i - T_s P^c_i$. Using this new energy harvesting process, the proposed framework can be used to obtain policies adapted to a constant circuit power consumption.
\end{remark}

Our goal is to reconstruct at the FC the sequence of measurements $\{x_i^k\}_{k=1}^{n}$ of each source in up to $d$ time slots since they were collected. Due to the continuous-valued nature of the sources and the rate constraint \eqref{eq:CapacityConstraint}, the reconstructed measurements  $\{\hat{x}_i^k\}_{k=1}^{n}$ will be unavoidably subject to some distortion. Such distortion will be characterized by a Mean Squared Error (MSE) metric:
\begin{align}
D_{i}=\frac{1}{n}\sum\limits_{k=1}^{n}\left(x^k_i - \hat{x}^k_i\right)^2,       \quad\quad\quad i=1,\dots,K.
\label{eq:distMetric}
\end{align}

\subsection{Source Coding and Distortion}

Hereinafter, we assume separability of source and channel coding at the sensor node. Hence, $\{x_i^k\}_{k=1}^{n}$ can be first encoded into a length-$n$ codeword (with a sufficiently large $n$) given by $\{u_i^k\}_{k=1}^{n}$. This process, as in \cite{ishwar2005rate}, can be modeled as
\begin{align}
u_{i}&=x_{i}+z_{i},     \quad\quad\quad i=1,\dots,K,
\label{eq:auxVar}
\end{align}
where $z_{i}$ denotes i.i.d. zero-mean Gaussian random noise of variance $\sigma_{z_i}^2$, which plays the role of encoding noise (the sample index has been omitted here for brevity). We know that, in order to decode the received data, the FC will exploit the available side information (i.e., all the preceding $u_{i}$). Hence, the sum of the (average) encoding rates per sample for the $i$-th source over the $d$ consecutive time slots must satisfy \cite{cover2012elements}
\begin{align}
\sum\limits_{j=i}^{i+d-1}R_{j,i}&  \geq I(x_i ; u_i | u_1 ,\dots, u_{i-1}),     \quad i=1,\dots,K,
\label{eq:Rate1}
\end{align}
where $I(\cdot;\cdot|\cdot)$ stands for the conditional mutual information. From \eqref{eq:auxVar}, this last expression can be rewritten as
\begin{align}
I(x_i ; u_i | u_1 ,\dots, u_{i-1}) = & H\left(u_i|u_1,\dots,u_{i-1}\right)  - \nonumber\\
                      & H\left(u_i|u_1,\dots,u_{i-1},x_i\right) \nonumber\\
     = &\log\left(1+\frac{\sigma^2_{x_i|u_1,\dots,u_{i-1}}}{\sigma^2_{z_i}}\right),
\end{align}
with $H(\cdot|\cdot)$ standing for the conditional entropy and $\sigma^2_{x_i|u_1,\dots,u_{i-1}}$  for the conditional variance of the $i$-th observation given all the previous data available at the FC. Hence, by taking equality in \eqref{eq:Rate1}, the variance of the encoding noise reads
\begin{align}
\sigma^2_{z_i}=\frac{\sigma^2_{x_i|u_1,\dots,u_{i-1}}}{\mathrm{e}^{\textstyle\sum\limits_{j=i}^{i+d-1}R_{j,i}}-1}.
\label{eq:varEncodingNoise}
\end{align}
In each time slot, the FC produces an optimal Minimum Mean Squared Error (MMSE) estimate of the observations which, as discussed earlier, exploits all the preceding $u_{i}$, namely
\begin{align}
 \hat{x}_i=\mathbb{E}\left[x_i|u_1,\dots,u_{i}\right], \quad\quad i=1,\dots, K.
\label{eq:decoding_MSE}
\end{align}
The distortion in the reconstruction of $x_i$ thus reads:
\begin{align}
D_i = \sigma^2_{x_i|u_1,\dots,u_{i}},
\end{align}
which, in turn, can be expressed as\footnote{With some abuse of notation, in the summation interval we write $k+d-1$. Still, we restrict such summations to the valid range of timeslot values, namely, $\max\{k+d-1,K\}$.} (see Appendix \ref{app:Distortion}, for a detailed derivation)
\begin{align}
D_i     =\sigma^2_x
        \biggl(
        \left(1-\rho\right)\sum^{i}_{j=2}
        &\rho^{i-j}
        \mathrm{e}^{-\textstyle\sum_{k=j}^{i}\sum_{l=k}^{k+d-1}R_{l,k}}
        \biggr.\nonumber\\
         \biggl. +&\rho^{i-1}
        \mathrm{e}^{-\textstyle\sum_{k=1}^{i}\sum_{l=k}^{k+d-1}R_{l,k}}
        \biggr).
        \label{eq:iDist}
\end{align}

\section{Minimization of the Average Distortion: Delay-Constrained Scenario}
\label{sec:MinDistd1}

Here, samples must be encoded, transmitted and reconstructed before the next time slot starts (i.e., $d=1$). The particularization of the channel capacity constraint \eqref{eq:CapacityConstraint} thus reads
\begin{align}
R_i \leq \log\left(1+|h_i|^2 p_i\right),
\label{eq:CapacityConstraintd1}
\end{align}
where $R_{i}$ stands for the transmission rate which is assigned to the $i$-th source in the $i$-th time slot \emph{only} (i.e., no summation of rates over subsequent time slots). Likewise, the rate-distortion constraint \eqref{eq:Rate1} can be particularized to
\begin{align}
R_{i}&  \geq I(x_i ; u_i | u_1 ,\dots, u_{i-1}),
\label{eq:Rated1}
\end{align}
From all the above, the reconstruction distortion in
\eqref{eq:iDist} simplifies to
\begin{align}
D_i     =\sigma^2_x
        \biggl(
        \left(1-\rho\right)\sum^{i}_{j=2}
        \rho^{i-j}
        \mathrm{e}^{-\textstyle\sum\limits_{k=j}^{i}R_{k}}
        +\rho^{i-1}
        \mathrm{e}^{-\textstyle\sum\limits_{k=1}^{i}R_{k}}
        \biggr).
        \label{eq:iDistd1}
\end{align}
Our goal is to find the optimal power $\{p_i\}$ and rate $\{R_i\}$ allocation that minimize the average distortion given by \eqref{eq:iDistd1} subject to the energy causality constraint of \eqref{eq:ECC} and the capacity constraint of \eqref{eq:CapacityConstraint}. Unfortunately, due to the coupling (over time slots) of the rates in the exponential terms of \eqref{eq:iDistd1}, this optimization problem cannot be solved analytically. To circumvent this, we define the \textit{cumulative rates} $r_{ij}$ as $r_{ij}\triangleq\sum_{k=j}^{i}R_{k}$, for $i=1,\dots,K,j=1,\dots,i$. By doing so, the optimization problem can be posed as:
\begin{subequations}
\begin{align}
   \min_{\substack{\{p_i\},\\ \{R_i\}, \\\{r_{ij}\} }}              \quad   &   \frac{\sigma^2_x}{K}
                                \sum^{K}_{i=1}
                                \biggl(
                                \left(1-\rho\right)\sum^{i}_{j=2}
                                \rho^{i-j}
                                \mathrm{e}^{-r_{ij}}
                                +
                                \rho^{i-1}
                                \mathrm{e}^{-r_{i1}}
                                \biggr)
                \label{eq:Distd1}\\
   \mathrm{s.t.}        \quad   &  r_{ij}=\sum\limits_{k=j}^{i}R_{k},                \quad i=1,\dots,K, j=1,\dots,i \label{eq:DistRateDefd1}\\
                \quad &  R_i \leq \log\left(1+|h_i|^2 p_i\right),           \quad i=1,\dots,K, \label{eq:DistRateCapd1}\\
                \quad &  \sum_{j=1}^{i}p_{j}\leq\sum_{j=1}^{i}E_{j},        \quad i=1,\dots,K, \label{eq:DistEnergyCausalityd1}\\
                \quad &  -p_{i} \leq 0,                         \quad i=1,\dots,K, \label{eq:PositivePowerd1}\\
                \quad &  -R_{i} \leq 0,                         \quad i=1,\dots,K, \label{eq:PositiveRated1} \\
                \quad &  -r_{ij} \leq 0,                            \quad i=1,\dots,K, j=1,\dots,i \label{eq:PositiveCumulativeRated1}
\end{align}
\label{eq:OptProblemd1}
\end{subequations}
where the optimization is with respect to variables $\{p_i\}$, $\{R_i\}$ and, also, $\{r_{ij}\}$ (this follows from the introduction of the additional constraint \eqref{eq:DistRateDefd1} associated to the definition of cumulative rates). Since the objective function \eqref{eq:Distd1} is convex and the constraints \eqref{eq:DistRateDefd1}-\eqref{eq:PositiveCumulativeRated1} define a convex feasible set, the optimization problem \eqref{eq:OptProblemd1} is convex and, thus, has a global solution \cite{boyd2009convex}. By satisfying the Karush-Kuhn-Tucker (KKT) conditions, we identify the necessary and sufficient conditions for optimality. The Lagrangian of \eqref{eq:OptProblemd1} reads
\begin{align}
\mathcal{L}     &=\frac{\sigma^2_x}{K}\sum^{K}_{i=1}
                \left(
                \left(1-\rho\right)\sum^{i}_{j=2}\rho^{i-j}\mathrm{e}^{-r_{ij}}+\rho^{i-1}\mathrm{e}^{-r_{i1}}
                \right)                                                                 \nonumber\\
        &\quad  + \sum\limits_{i=1}^{K}\sum\limits_{j=1}^{i}\lambda_{ij}\left(r_{ij} - \sum\limits_{k=j}^{i}R_{k}\right)        \nonumber\\
        &\quad  + \sum\limits_{i=1}^{K}\mu_{i}\left(R_i - \log\left(1+|h_i|^2 p_i\right)\right)                     \nonumber\\
        &\quad  + \sum\limits_{i=1}^{K}\beta_{i}\left(\sum_{j=1}^{i}p_{j} - \sum_{j=1}^{i}E_{j}\right)              \nonumber\\
        &\quad  - \sum\limits_{i=1}^{K}\eta_{i}p_{i} - \sum\limits_{i=1}^{K}\phi_{i}R_{i} - \sum\limits_{i=1}^{K}\sum\limits_{j=1}^{i}\delta_{i}r_{ij},
\end{align}
where $\{\mu_i\} \geq 0$, $\{\beta_i\} \geq 0$, $\{\eta_i\} \geq 0$, $\{\phi_{i}\} \geq 0$, $\{\delta_{ij}\} \geq 0$ and $\{\lambda_{ij}\}$ stand for the corresponding Lagrange multipliers. The additional complementary slackness conditions are given by
\begin{align}
    \mu_{i}\left(R_i - \log\left(1+|h_i|^2 p_i\right)\right)=0,         &\quad \forall i,   \label{eq:slackIndividualRatesd1}\\
    \beta_{i}\left(\sum_{j=1}^{i}p_{j} - \sum_{j=1}^{i}E_{j}\right)=0,  &\quad \forall i,   \label{eq:slackECCd1}\\
    \eta_{i}p_{i}=0,    &\quad \forall i,                                   \label{eq:slackPowerConstraintd1}\\
    \phi_{i}R_{i}=0,    &\quad \forall i,                                   \label{eq:slackRateConstraintd1} \\
    \delta_{i}r_{ij}=0, &\quad \forall i, j.                                    \label{eq:slackRateConstraintd1}
\end{align}
Finally, by taking the derivative of the Lagrangian with respect to $p_i$, $R_i$, $r_{ij}$ and letting them be equal to zero we the set of stationarity conditions follow, namely,
\begin{equation}
    \frac{\partial\mathcal{L}}{\partial p_i} = - \frac{\mu_i |h_i|^2}{1+|h_i|^2 p_i}+\sum\limits_{j=i}^{K}\beta_{j}-\eta_i=0,
    \label{eq:powerLagrangiand1}
\end{equation}
\begin{equation}
    \frac{\partial\mathcal{L}}{\partial R_i} = -\sum\limits_{k=i}^{K}\sum\limits_{j=1}^{i}\lambda_{kj}+\mu_{i}-\phi_{i}=0,
    \label{eq:rateMultipliersd1}
\end{equation}
\begin{equation}
\frac{\partial\mathcal{L}}{\partial r_{ij}}=
\begin{cases}
   -\frac{\sigma^2_x}{K}\rho^{i-j}\mathrm{e}^{-r_{ij}}+\lambda_{ij}-\delta_{ij}=0,                 &   \text{if } j=1, \\
   -\frac{\sigma^2_x}{K}\left(1-\rho\right)\rho^{i-j}\mathrm{e}^{-r_{ij}}+\lambda_{ij}-\delta_{ij}=0,      &   \text{if } j\neq 1 .
\end{cases}
    \label{eq:rateLagrangiand1}
\end{equation}

\subsection{Optimal Power Allocation}

\begin{figure}[t!]
    \centering
    \subfigure[Two-dimensional directional waterfilling.]
    {
        \includegraphics[width=0.8\columnwidth]{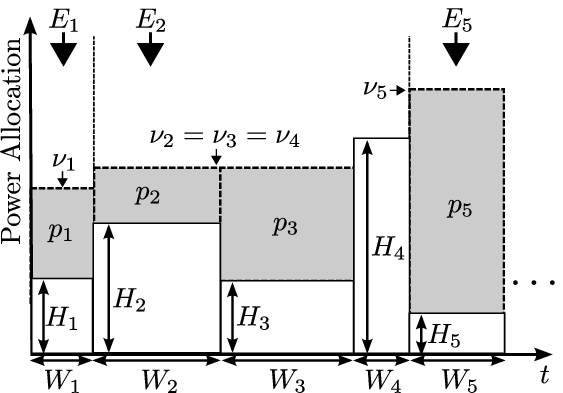}
        \label{fig:WaterFilling}
    }
    \subfigure[Reverse waterfilling with multiple waterlevels.]
    {
        \includegraphics[width=0.8\columnwidth]{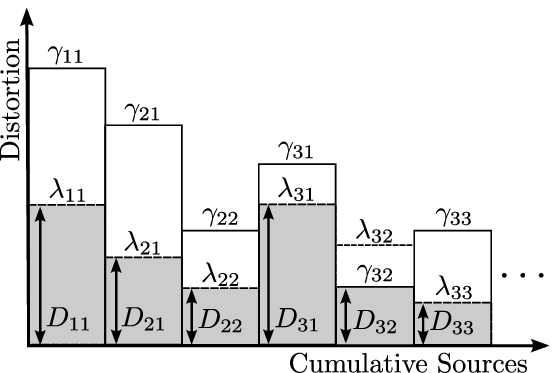}
        \label{fig:ReverseWaterFilling}
    }
    \caption{Optimal power and cumulative rate allocation.}
    \label{fig:WFandRWFInterpretation}
\end{figure}

From the stationarity conditions on $p_{i}$ \eqref{eq:powerLagrangiand1} and $R_i$ \eqref{eq:rateMultipliersd1}, and the slackness conditions of \eqref{eq:slackRateConstraintd1}, the optimal power allocation follows:
\begin{equation}
    p_{i}^{\star}=\left[\frac{\sum\limits_{k=i}^{K}\sum\limits_{j=1}^{i}\lambda_{kj}}{\sum\limits_{j=i}^{K}\beta_{j}}-\frac{1}{|h_i|^2}\right]^{+},
    \quad\quad i=1,\dots,K,
    \label{eq:optPowerAllocationd1}
\end{equation}
where $[ \cdot ]^{+}=\max\{\cdot,0\}$. This solution can be interpreted as a two-dimensional directional waterfilling, as shown in Fig. \ref{fig:WaterFilling}. For each time slot $i$, we have a rectangle of solid material of width $W_i \triangleq\sum_{k=i}^{K}\sum_{j=1}^{i}\lambda_{kj}$ and height $H_i \triangleq 1\big/\left(|h_i|^2\sum_{k=i}^{K}\sum_{j=1}^{i}\lambda_{kj}\right)$. Right-permeable taps are placed in time slots with energy arrivals. Water is consequently poured up to a waterlevel $\nu_i \triangleq 1/\sum_{j=i}^{K}\beta_{j}$. The resulting power allocation corresponds to the area of water above the solid rectangle.

%%%%%%%%%%%%%%%%%%%%%%%%%%%%%%%%%%%%%%%%%%
\begin{figure*}[t!]
\newcounter{tmpCounter} %Fix for equation numbers
\setcounter{tmpCounter}{\value{equation}}  %Current equation number
\setcounter{equation}{29} %Equation number - 1
\begin{equation}
    \left\{
    \begin{aligned}
    r_{ii}=\sum_{j=i}^{i+d-1}R_{j,i}                        ,i=1,\ldots,K                       \\
    \sum\limits_{j=i-d+1}^{i}R_{i,j} \leq \log\left(1+|h_i|^2 p_i\right)    ,i=1,\ldots,K                       \\
    R_{i,j}\geq 0                                       ,i=1,\ldots,K, j=i-d+1,\ldots,i
    \end{aligned}
    \right\}
\equiv
    \left\{
    \begin{aligned}
     r_{ij} =\sum\limits_{k=j}^{i}r_{kk}                        ,i=1,\ldots,K, j=1,\ldots,i-1\\
     r_{ij} \leq \sum\limits_{k=j}^{d+i-1}\log\left(1+|h_k|^2 p_k\right)    ,i=1,\ldots,K, j=1,\ldots,i\\
    r_{ij} \geq 0                                       ,i=1,\ldots,K, j=1,\ldots,i
    \end{aligned}
    \right\}
    \label{eq:EquivalentSystem}
\end{equation}
\setcounter{equation}{\value{tmpCounter}} %Restore equation number
\hrulefill
\end{figure*}
%%%%%%%%%%%%%%%%%%%%%%%%%%%%%%%%%%%%%%%%%%%%%

\subsection{Optimal Rate Allocation}
Next, by solving \eqref{eq:rateLagrangiand1} for $r_{ij}$, and taking into account the corresponding slackness conditions, the optimal cumulative rate allocation can be written as
\begin{equation}
r_{ij}^{\star}=
\begin{cases}
    \left[\log\left(\dfrac{\frac{1}{K}\sigma^2_x\rho^{i-j}}{\lambda_{ij}}\right)\right]^{+},                &   \text{if } j=1, \\
    \left[\log\left(\dfrac{\frac{1}{K}\sigma^2_x\left(1-\rho\right)\rho^{i-j}}{\lambda_{ij}}\right)\right]^{+}, &   \text{if } j\neq 1.
    \label{eq:optCumulativeRatesd1}
\end{cases}
\end{equation}
From this last expression, it becomes apparent that, necessarily, $\{\lambda_{ij}\}>0$. Hence, from \eqref{eq:rateMultipliersd1}, we have that $\{\mu_{i}\}>0$ too. This implies that constraint \eqref{eq:DistRateCapd1} is satisfied with equality. Moreover, expression \eqref{eq:optCumulativeRatesd1} can be readily interpreted in terms of a \emph{reverse} water-filling solution for the reconstruction of parallel Gaussian sources \cite[Chapter 10]{cover2012elements}. To see that, we define
\begin{equation}
\gamma_{ij}=
\begin{cases}
    \frac{1}{K}\sigma^2_x\rho^{i-j},                    &   \text{if } j=1, \\
    \frac{1}{K}\sigma^2_x\left(1-\rho\right)\rho^{i-j},     &   \text{if } j\neq 1,
\label{eq:gamma_ii_defd1}
\end{cases}
\end{equation}
and
\begin{equation}
D_{ij}=
\begin{cases}
    \lambda_{ij},               &   \text{if } \lambda_{ij}<\gamma_{ij}, \\
    \gamma_{ij},            &   \text{if } \lambda_{ij}\geq\gamma_{ij}.
\end{cases}
\end{equation}
Bearing the above in mind, equation \eqref{eq:optCumulativeRatesd1} can be rewritten as
\begin{equation}
r_{ij}^{\star}=\left[\log\left(\frac{\gamma_{ij}}{D_{ij}}\right)\right]^+.
\label{eq:optCumulativeRatesSimplifiedd1}
\end{equation}
As Figure \ref{fig:ReverseWaterFilling} illustrates, this solution mimics that of a rate-distortion allocation problem for parallel Gaussian sources. However, here the allocated rates $r_{ij}^{\star}$ (and sources) are \textit{cumulative} rather than \textit{individual}; and the reverse water level given by $\lambda_{ij}$ is not constant. Besides, the numerator in the argument of \eqref{eq:optCumulativeRatesSimplifiedd1} does not only depend on the variance of the sources $\sigma^2_x$ but also on the correlation coefficient $\rho$, as \eqref{eq:gamma_ii_defd1} evidences.

Finally, by replacing \eqref{eq:optCumulativeRatesSimplifiedd1} in \eqref{eq:iDistd1}, the optimal distortion for the reconstruction of the $i$-th source reads
\begin{equation}
    D_{i}^{\star}=\sum_{j=1}^i D_{ij}.
\end{equation}
that is, it can be computed as the sum of the distortions associated to the corresponding cumulative rates.

\subsection{Optimization Algorithm}
\label{sec:MinDistd1Alg}

\begin{algorithm}[t]
    \caption{Optimal power and rate allocation for the delay-constrained case.}
    \label{alg:OptAlgorithmd1}
    \begin{algorithmic}[1]
        \State \textbf{Initialize:} $\{\lambda_{ij}^{(t)}\}:=0$.
        \State \textbf{Step 1:} For all $i$, allocate power.
    \State $p_{i}^{(t+1)}  := \left[\dfrac{\sum\limits_{k=i}^{K}\sum\limits_{j=1}^{i}\lambda_{kj}^{(t)}}
        {\sum\limits_{j=i}^{K}\beta_j}-\dfrac{1}{|h_i|^2}\right]^{+}$
        \State \textbf{Step 2:}  For all $i, j$, cumulative rate allocation.
            \State $r_{ij}^{(t+1)}  := \left[\log\left(\dfrac{\gamma_{ij}}{\lambda_{ij}^{(t)}}\right)\right]^{+}$
        \State \textbf{Step 3:}  For all $i, j$, update multiplier.
            \State $\lambda_{ij}^{(t+1)}  :=
            \biggl[\lambda_{ij}^{(t)}+ \alpha\biggl(r_{ij}^{(t+1)} - \sum\limits_{k=j}^{i}\log\left(1+|h_k|^2 p_k^{(t+1)}\right)\biggr)\biggr]
            $
        \State \textbf{Step 4:} Go to Step 1 until stopping criteria is met.
    \end{algorithmic}
\end{algorithm}

As discussed in the previous section, the optimal power \eqref{eq:optPowerAllocationd1} and cumulative rate \eqref{eq:optCumulativeRatesd1} allocation are coupled\footnote{When dealing with hybrid power supplies (power grid and energy harvesting), two-stage waterfilling structures have also been identified in \cite{gong2013optimal, hu2015optimal}.} by the Lagrange multipliers $\lambda_{ij}$. Further, one can easily prove that problem \eqref{eq:OptProblemd1} satisfies Slater's condition, and therefore, strong duality holds \cite{boyd2009convex}. Since in these conditions the duality gap is zero, we propose to solve the corresponding \emph{dual} problem in order to determine the \emph{primal} solution (power and rates) in which we are interested. To that aim, we resort to the subgradient method \cite{bertsekas1999nonlinear} on which basis the solution to the dual problem, $\{\lambda_{ij}\}$, can be iteratively found (convergence can be guaranteed under some mild conditions). Specifically, in the $t$-th iteration, the Lagrange multipliers are updated as follows\footnote{Here, we use extended-value definitions for all functions \cite{boyd2009convex}, thus taking $+\infty$ values outside their respective domain.}:
\begin{equation}
        \lambda_{ij}^{(t+1)} := \biggl[\lambda_{ij}^{(t)}+\alpha\biggl(r_{ij}^{(t+1)} -
                     \sum\limits_{k=j}^{i}\log\left(1+|h_k|^2 p_k^{(t+1)}\right)\biggr)\biggr],
                     \label{eq:update_dual}
\end{equation}
with $\alpha$ standing for the corresponding step size. In Algorithm \ref{alg:OptAlgorithmd1}, we summarize the proposed procedure to solve the power and cumulative rate allocation problem.

The algorithm corresponds to a a subgradient ascent on the dual function. Hence, it has a convergence rate of the order of $\mathcal{O}(1/\sqrt{t})$ \cite[Chapter 8.2]{bertsekas2003convex}. Moreover, for a node, a single iteration of the algorithm will be of the order of $\mathcal{O}(K \log K)$, as it is a form of waterfilling. Namely, sorting takes $\mathcal{O}(K \log K)$ operations, while each waterfilling operation takes $\mathcal{O}(K)$ operations and there are at most $\mathcal{O}(\log K)$ waterfillings to be done, since we can compute a binary search between the water bins.

\section{Minimization of the Average Distortion: Delay-Tolerant Scenario}
\label{sec:MinDist}

Here, we address the more general case in which data is allowed to be transmitted and reconstructed within $d>1$ time slots after samples are collect and encoded. Again, to render the problem solvable, we define the \textit{cumulative rates} as
\begin{equation}
r_{ij}\triangleq \sum_{k=j}^{i}\sum_{l=k}^{k+d-1}R_{l,k}
\quad \text{for }i=1,\ldots,K,j=1,\ldots,i.
\label{eq:defCumulativeRatesPre}
\end{equation}
In delay-tolerant scenarios, each time slot conveys data from up to $d$ different sources (see Fig. \ref{fig:systemModelTX}). Hence, the number of unknowns ($\{R_{i,j}\}$), $Kd - d(d-1)/2$ in total, exceeds the number of equations given by the capacity constraints \eqref{eq:CapacityConstraint}, $K$ in total. Consequently, the system of equations becomes underdetermined. This means that, even if a unique solution exists when optimizing on the \emph{cumulative} rates $r_{ij}$ (as we discuss next), there exist \emph{multiple} solutions for the \emph{individual} rates. Thus, we propose to solve the optimization problem in terms of cumulative rates (only), and then define some criteria to select one solution in terms of individual rates (this will be further elaborated in Section \ref{sec:NumRes} ahead).

To start with, we need to rewrite not only (i) the objective function given by \eqref{eq:iDist}; but, also, (ii) the set of constraints, in terms of cumulative rates. The latter can be accomplished by expressing the cumulative rates in the following recursive form:
\begin{equation}
r_{ij}\triangleq
\begin{cases}
    \sum_{k=j}^{i}r_{kk},           &   \text{if } j\neq i \\
    \sum_{k=j}^{j+d-1}R_{k,j},      &   \text{if } j=i
\end{cases}
\label{eq:defCumulativeRates}
\end{equation}
for $i=1,\ldots,K,j=1,\ldots,i$, and then resorting to Fourier-Motzkin elimination \cite{schrijver1998theory}.  Further, we prove that the system obtained by Fourier-Motzkin elimination is equivalent.

\begin{proposition}
\label{prop:equivalence}
The systems of inequalities in \eqref{eq:EquivalentSystem}  are equivalent when solving optimization problem \eqref{eq:OptProblem}. That is, the set of variables $\{p_i\}$,$\{R_{i,j}\}$ and $\{r_{ij}\}$ satisfy the constraints on the left hand side of \eqref{eq:EquivalentSystem} if and only if they satisfy the constraints on the right hand side of \eqref{eq:EquivalentSystem}.
\end{proposition}
\begin{proof}
See Appendix \ref{app:Equivalence}.
\end{proof}

Finally, in order to pose the optimization problem, it suffices to include the corresponding \emph{energy harvesting} constraints of \eqref{eq:ECC} too, namely
\begin{subequations}
\begin{align}
   \min_{\substack{\{p_i\}, \\\{r_{ij}\} }}     \quad   &   \frac{\sigma^2_x}{K}
                                    \sum^{K}_{i=1}
                                    \biggl(
                                    \left(1-\rho\right)\sum^{i}_{j=2}
                                    \rho^{i-j}
                                    \mathrm{e}^{-r_{ij}}
                                    +
                                    \rho^{i-1}
                                    \mathrm{e}^{-r_{i1}}
                                    \biggr)
                \label{eq:Dist}\\
   \mathrm{s.t.}        \quad   &  r_{ij}=\sum\limits_{k=j}^{i}r_{kk},                       i=1,\dots,K, j=1,\dots,i-1 \label{eq:DistRateDef}\\
                \quad &  r_{ij} \leq \sum\limits_{k=j}^{d+i-1}\log\left(1+|h_k|^2 p_k\right),   \nonumber\\
                    & i=1,\dots,K, j=1,\dots,i \label{eq:DistRateCap}\\
                \quad &  \sum_{j=1}^{i}p_{j}\leq\sum_{j=1}^{i}E_{j},                \quad i=1,\dots,K, \label{eq:DistEnergyCausality}\\
                \quad &  -p_{i} \leq 0,                                 \quad i=1,\dots,K, \label{eq:PositivePower}\\
                \quad &  -r_{ij} \leq 0,                                    \quad i=1,\dots,K, j=1,\dots,i  \label{eq:PositiveRate}
\end{align}
\label{eq:OptProblem}
\end{subequations}
where, clearly, the optimization is now with respect to variables $\{p_i\}$ and $\{r_{ij}\}$. Differently from Section \ref{sec:MinDistd1}, constraint \eqref{eq:DistRateDef} guarantees, on the one hand, that the cumulative rates satisfy definition \eqref{eq:defCumulativeRates}. On the other, constraint \eqref{eq:DistRateCap} enforces the cumulative rates to satisfy the per time slot channel capacity constraint.

The optimization problem \eqref{eq:OptProblem} is convex and can be solved in closed-form by (i) computing the Lagrangian function:
\begin{align}
\mathcal{L}     &=\frac{\sigma^2_x}{K}\sum^{K}_{i=1}
                \left(
                \left(1-\rho\right)\sum^{i}_{j=2}\rho^{i-j}\mathrm{e}^{-r_{ij}}+\rho^{i-1}\mathrm{e}^{-r_{i1}}
                \right)                                                             \nonumber\\
        &\quad  + \sum\limits_{i=1}^{K}\sum\limits_{j=1}^{i}\mu_{ij}\left(r_{ij} - \sum\limits_{k=j}^{i}r_{kk}\right)       \nonumber\\
        &\quad  + \sum\limits_{i=1}^{K}\sum\limits_{j=1}^{i}\lambda_{ij}\left(r_{ij} - \sum\limits_{k=j}^{d+i-1}\log\left(1+|h_i|^2 p_i\right)\right) \nonumber\\
        &\quad  + \sum\limits_{i=1}^{K}\beta_{i}\left(\sum_{j=1}^{i}p_{j} - \sum_{j=1}^{i}E_{j}\right)              \nonumber\\
        &\quad  - \sum\limits_{i=1}^{K}\eta_{i}p_{i} - \sum\limits_{i=1}^{K}\sum\limits_{j=1}^{i}\delta_{ij}r_{ij}.
\end{align}
with $\{\lambda_{ij}\} \geq 0$, $\{\beta_i\} \geq 0$, $\{\eta_i\} \geq 0$, $\{\delta_{i}\} \geq 0$; and $\{\mu_{ij}\}$ standing for the corresponding Lagrangian multipliers; and (ii) satisfying the Karush-Kuhn-Tucker (KKT) conditions that follow from the Lagrangian. Along the lines of Section \ref{sec:MinDistd1}, the optimal power allocation reads,
\begin{equation}
    p_{i}^{\star}=\left[\frac{\sum\limits_{k=i-d+1}^{K}\sum\limits_{l=1}^{k}\lambda_{kl}}{\sum\limits_{j=i}^{K}\beta_{j}}-\frac{1}{|h_i|^2}\right]^{+},
    \quad i=1,\dots,K.
    \label{eq:optPowerAllocation}
\end{equation}
The optimal power allocation for the (more general) delay-tolerant scenario admits again a two-dimensional directional waterfilling interpretation. Differently from the delay-constrained scenario, the width and height of the solid rectangle, namely $W_i \triangleq \sum_{k=i-d+1}^{K}\sum_{l=1}^{k}\lambda_{kl}$ and $H_i \triangleq 1\big/\left(|h_i|^2\sum_{k=i-d+1}^{K}\sum_{l=1}^{k}\lambda_{kl}\right)$, have an explicit dependence on $d$, the maximum latency\footnote{As expected, these expressions simplify to the ones for the delay-constrained scenario for $d=1$.}.

Along the lines of the preceding section, the optimal cumulative rates for the delay-constrained scenario follow:
\begin{equation}
r_{ij}^{\star}=\left[\log\left(\dfrac{\gamma_{ij}}{\lambda_{ij}+\bar{\mu}_{ij}}\right)\right]^{+}.
    \label{eq:optCumulativeRates}
\end{equation}
where we have defined
\begin{equation}
\bar{\mu}_{ij}=
\begin{cases}
    -\sum\limits_{k=i}^{K}\sum\limits_{\substack{l=1 \\ l\neq k}}^{i}\mu_{kl},          &   \text{if } i=j, \\
    \mu_{ij},                                                   &   \text{if } i\neq j.
\label{eq:gamma_ii_def}
\end{cases}
\end{equation}
and $\gamma_{ij}$ is given by \eqref{eq:gamma_ii_defd1}. Again, this solution can be interpreted in terms of a classical reverse waterfilling scheme.
\begin{algorithm}[t]
    \caption{Optimal power and rate allocation for the delay-tolerant case.}
    \label{alg:OptAlgorithm}
    \begin{algorithmic}[1]
        \State \textbf{Initialize:} $\{\lambda_{ij}^{(t)}\}:=0$, $\{\mu_{ij}^{(t)}\}:=0$.
        \State \textbf{Step 1:} For all $i$, allocate power.
    \State $p_{i}^{(t+1)} := \left[\dfrac{\sum\limits_{k=i-d+1}^{K}\sum\limits_{l=1}^{k}\lambda_{kl}^{(t)}}
        {\sum\limits_{j=i}^{K}\beta_j}-\dfrac{1}{|h_i|^2}\right]^{+}$
        \State \textbf{Step 2:}  For all $i, j$, cumulative rate allocation.
            \State $r_{ij}^{(t+1)} := \left[\log\left(\dfrac{\gamma_{ij}}{\lambda_{ij}^{(t)}+\bar{\mu}_{ij}^{(t)}}\right)\right]^{+}$
        \State \textbf{Step 3:}  For all $i, j$, update multipliers.
            \State $\lambda_{ij}^{(t+1)} :=
            \biggl[\lambda_{ij}^{(t)}+\alpha_{\lambda}\biggl(r_{ij}^{(t+1)} - \biggr.\biggr.$ \\
            $\quad\quad\quad\quad\quad \biggl.\biggl. \sum\limits_{k=j}^{d+i-1}\log\left(1+|h_k|^2 p_k^{(t+1)}\right)\biggr)\biggr]^{+}$
            \State $\mu_{ij}^{(t+1)} :=
            \biggl[\mu_{ij}^{(t)}+\alpha_{\mu}\biggl(r_{ij}^{(t+1)} - \sum\limits_{k=j}^{i}r_{kk}^{(t+1)}\biggr)\biggr]$
        \State \textbf{Step 4:} Go to Step 1 until stopping criteria is met.
    \end{algorithmic}
\end{algorithm}

As in Section \ref{sec:MinDistd1Alg}, we solve the corresponding dual problem by resorting to the subgradient method. However, now \emph{both} dual variables $\lambda_{ij}$ and $\mu_{ij}$ must be updated as follows
\begin{equation}
        \lambda_{ij}^{(t+1)} := \biggl[\lambda_{ij}^{(t)}+\alpha_{\lambda}\biggl(r_{ij}^{(t+1)} - \sum\limits_{k=j}^{d+i-1}\log\left(1+|h_k|^2 p_k^{(t+1)}\right)\biggr)\biggr]^{+},
\end{equation}
\begin{equation}
        \mu_{ij}^{(t+1)} := \biggl[\mu_{ij}^{(t)}+\alpha_{\mu}\biggl(r_{ij}^{(t+1)} - \sum\limits_{k=j}^{i}r_{kk}^{(t+1)}\biggr)\biggr],
\end{equation}
where $\alpha_{\lambda}$ and $\alpha_{\mu}$ denote the corresponding step sizes. Algorithm \ref{alg:OptAlgorithm} details the proposed procedure to obtain the optimal power and cumulative rate allocation.

\section{Numerical Results}
\label{sec:NumRes}

In this section, we assess the performance of the proposed optimal power and rate allocation schemes. We are particularly interested in analyzing the impact of the correlation coefficient $\rho$ and the delay $d$ in the resulting transmission policies. For this reason, in all numerical results we have set the channel gains to a (constant) unit value. Unless otherwise stated, the simulation setup considers a system with $K=10$ time slots and an (arbitrary) energy harvesting profile with energy arrivals given by $E_1=0.2$, $E_3=0.6$, $E_6=0.8$ and $E_7=1.4$.
\begin{figure}[t]
    \centering
    \includegraphics[width=\columnwidth]{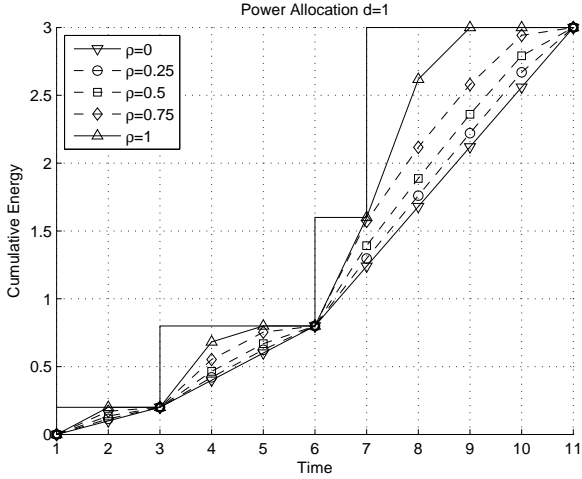} 
    \caption{Optimal power allocation for $d=1$ and varying correlation coefficient $\rho$.}
    \label{fig:PowerAllocationVarRhoD1}
\end{figure}

\subsection{Delay-Constrained Scenario ($d=1$)}

The resulting optimal power allocation policy is shown in Figure \ref{fig:PowerAllocationVarRhoD1}. For uncorrelated sources ($\rho=0$), the optimal policy turns out to be the well-known geometric solution of \cite{yang2012optimal} and \cite{zafer2009calculus}. This corresponds to the tightest string below the cumulative energy harvesting curve connecting the original and the total harvested energy by the end of time slot $K$. However, as the correlation increases, the harvested energy tends to be spent (i.e., allocated as transmit power) sooner. As a result, in Fig. \ref{fig:PowerAllocationVarRhoD1} the slope of the energy consumption curves right after new energy arrivals (e.g., in the beginning of time slot 3) increases with $\rho$. This indicates that, in order to minimize the average distortion, one should encode the observations as accurately as possible when some new energy is made available. This stems from the fact that past observations are used here as side information at the receiver. Intuitively, the earlier an observation is accurately encoded, the more estimates (in subsequent time slots) can benefit from such an increased accuracy. This holds true even at the expense of a reduced (or zero, as in time slot 10, for $\rho=1$) transmit power being allocated to some subsequent time slots. That is, at the expense of suspending data transmission. All the above is in stark contrast with the uncorrelated case studied in \cite{yang2012optimal} where transmit power is (i) strictly positive for all time slots and (ii) a monotonically increasing function.

\begin{figure}[t]
    \centering
\includegraphics[width=\columnwidth]{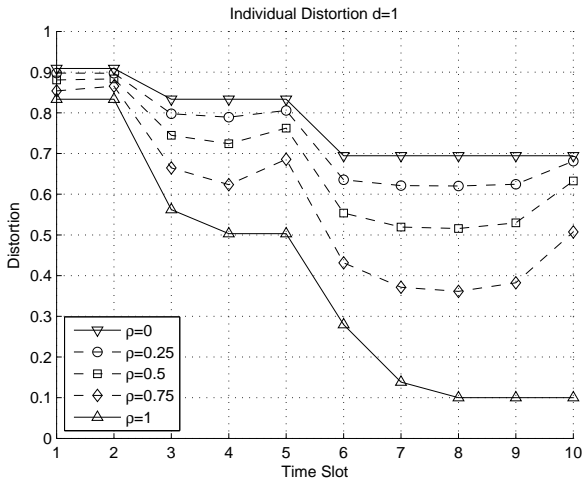}
    \caption{Individual distortion for $d=1$ and varying $\rho$.}
    \label{fig:IndividualDistortionVarRhoD1}
\end{figure}

Figure \ref{fig:IndividualDistortionVarRhoD1} depicts the reconstruction distortion for \emph{each} source (and time slot since $d=1$) associated to the optimal policy. Unsurprisingly, the higher the correlation, the more predictable the sources become and, hence, the lower the distortion (curves are shifted downwards). For correlated sources, however, distortion does not monotonically decrease with time slot index. As discussed in the previous paragraph, this stems from the \emph{anticipated} consumption of the harvested energy for the encoding of previous observations. Consequently, one can observe (i) a substantial decrease of the individual distortion for sources in time slots with energy arrivals (time slots 3, 6, and 7); and (ii) distortion upturns in time slots where the energy harvested so far has been spent or is close to (time slots 2, 5, and 10). Still, the average distortion is lower.

\subsection{Delay-Tolerant Scenario ($d>1$)}

Figure \ref{fig:PowerAllocationVarDRho08} illustrates the impact of delay on the optimal power allocation. Interestingly, as $d$ increases the solution converges to the tightest string below the cumulative energy harvesting curve of \cite{yang2012optimal}. The intuition behind is as follows. To recall, the tightest string solution attempts to maximize the total throughput (rate) for the whole transmission period. To that aim, the sequence of transmit powers (and rates) must be monotonically increasing, that is, transmit power is higher by the end of the transmission period (i.e, last time slot(s)). For $d=1$, on the contrary, the allocated transmit power (and, thus, rate) is higher in time slots with energy arrivals, and not necessarily in the last one(s). Moreover, the source must be reconstructed \emph{immediately}, that is, after $d=1$ time slots (assuming the processing time at the FC to be negligible). In other words, there is some \emph{urgency} to allocate power (namely, spend energy). This is in stark contrast with the tightest string solution where higher power and rates can be found at the \emph{end}. Things, however, are radically different when $d$ increases. On the one hand, the deadline by which individual sources must be reconstructed is shifted $d$ time slots towards the \emph{end}. On the other, the rates (and power) needed to encode a specific source can be allocated over \emph{multiple} time slots, rather than just one. Hence, for increasing $d$ the urgency to allocate power decreases and, thus, the way in which power is allocated is more aligned with that of the tightest string solution.
\begin{figure}[t]
    \centering
    \includegraphics[width=\columnwidth]{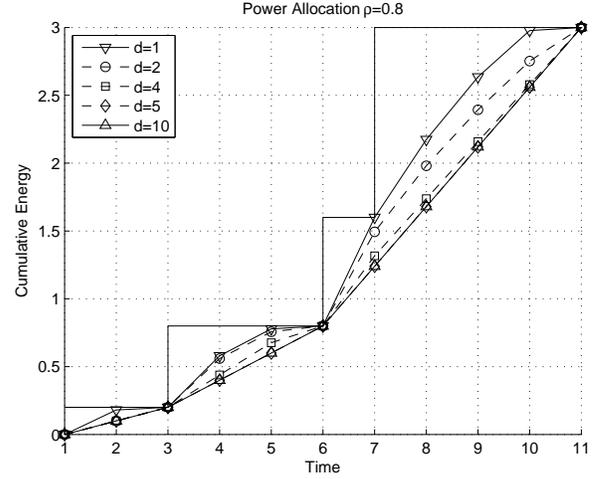}
    \caption{Optimal power allocation for $\rho=0.8$ and varying delay $d$.}
    \label{fig:PowerAllocationVarDRho08}
\end{figure}

Figure \ref{fig:AvgDistortion} depicts the average distortion as a function of delay. Clearly, the average distortion decreases with delay since the higher the delay, the higher the degrees of freedom to allocate transmit power (and, thus, spend energy in a more sensible manner). Unsurprisingly, distortion is lower for higher values of $\rho$, since the preceding (correlated) sources used as side information at the FC are more informative.
\begin{figure}[t]
    \centering
    \includegraphics[width=\columnwidth]{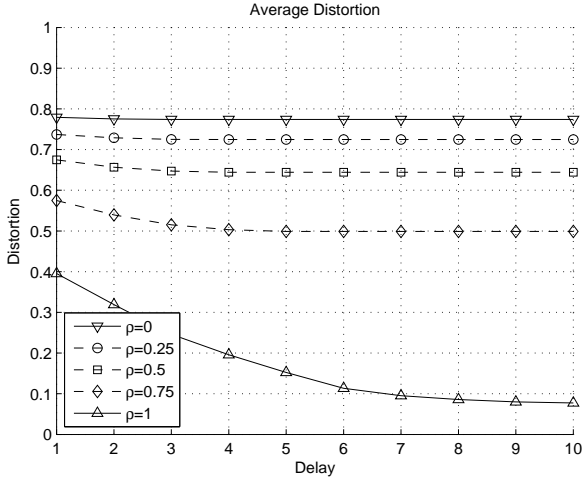}
    \caption{Average distortion vs. delay $d$ for varying correlation $\rho$.}
    \label{fig:AvgDistortion}
\end{figure}

Next, we investigate to what extent our system leverages on the knowledge on source correlation. To that aim, the rate and power allocation policy from \cite{orhan2013delay}, which was derived for a scenario with \emph{uncorrelated} sources, is used as a benchmark. Specifically, whereas the source encoding rate depends on $I(x_i ; u_i)$ (namely, the mutual information with the \emph{current} source only) the sources at the FC are reconstructed according to \eqref{eq:decoding_MSE}. Our approach, on the contrary, exploits correlation both in the encoding and decoding/reconstruction processes. Figure \ref{fig:distortionComparison} shows the normalized reduction (difference) in the average distortion attained by such benchmark and our scheme. For delay-constrained scenarios, the reduction in distortion can be as high $25\%$ for our scheme. For delay-tolerant ones ($d=10$), reduction can go up to $80\%$, which is very remarkable.

\begin{figure}[t]
    \centering

    \includegraphics[width=\columnwidth]{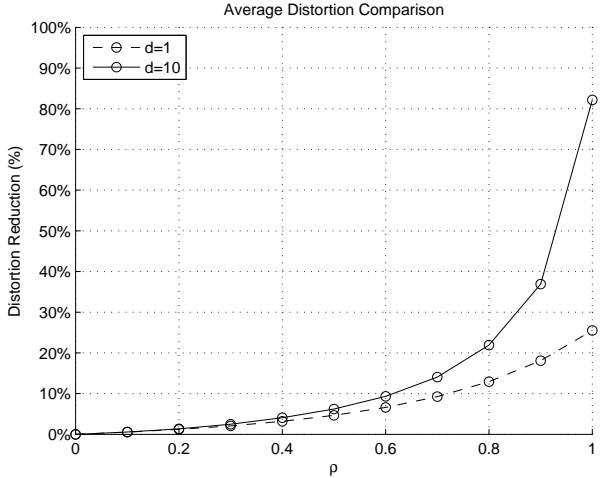}
    \caption{Reduction in average distortion.}
    \label{fig:distortionComparison}
\end{figure}

\subsection{Comparison with an Online Policy}
As discussed earlier, the proposed (offline) transmit power and rate allocation scheme requires \emph{non-causal} knowledge on energy arrivals. Here, instead, we introduce a more realistic \emph{online} version just requiring \emph{causal} knowledge. The offline scheme will be used as a benchmark.

Similar to \cite{ozel2011transmission}, a \emph{myopic}\footnote{More general online policies accounting for different degrees of availability of channel and energy state information can be also be considered (see e.g., \cite{li2015general}).} online policy can be computed as follows. Assume for a moment that, after harvesting some energy in the initial timeslot (i.e., $E_1>0$), no additional energy is harvested in subsequent timeslots. Hence, we let $E_2=\cdots=E_K=0$ and solve problem \eqref{eq:OptProblem} for $k=1,\ldots,K$. In the absence of knowledge on future energy arrivals, this is a sensible approach too. After all, distortion would be minimized should no additional energy be actually harvested. And, otherwise, we can react accordingly. Let $k_0<K$ denote the next timeslot in which some energy is harvested (i.e., $E_{k_0}>0$). For the preceding timeslots (i.e., $k=1,\ldots,k_0-1$), we force the power and rate allocations computed after the last energy arrival to remain unchanged. Hence, the unspent energy in the beginning of timeslot $k_0$ reads $E_{k_0}^u=\sum_{j=1}^{k_0-1}E_j-\sum_{j=1}^{k_0-1}p_j$. Next, we let $E_{k_0}  := E_{k_0}^u + E_{k_0}$ and $E_{k_0+1}=\cdots=E_K=0$ and, again, solve problem \eqref{eq:OptProblem} for $k=k_0,\ldots,K$. That is, we compute the optimal power and rate allocations for all subsequent timeslots. This procedure is iterated until all energy arrivals have been accounted for.

Of course, no optimality can be claimed for the resulting policy. Still, the interesting property of such scheme is its ability to adjust (re-compute) the remaining power and rate allocations every time that some energy is harvested. By doing so, the additional (and causal) knowledge on energy arrivals is effectively exploited.

Figure \ref{fig:online} illustrates the performance of the offline and online policies vs. the intensity rate of energy arrivals (which are modeled as a Poisson process). Unsurprisingly, the distortion of the offline versions turns out to be a lower bound of that attained by online ones. For a given intensity rate, the additional distortion associated to the online version can be regarded as moderate (some $20\%$ at an energy arrival rate equal to 1, $\rho=0.2$, and $d=1$). Interestingly, the online version requires a $40\%$ increase of the intensity rate to achieve the same distortion as its offline counterpart (for the same operating point). The distortion gap becomes narrower for delay-tolerant scenarios ($d=10$) and wider in percentage for scenarios with high correlation (see $\rho=0.8$ curves) or when the intensity rate increases.

\begin{figure}[t]
    \centering
    \includegraphics[width=\columnwidth]{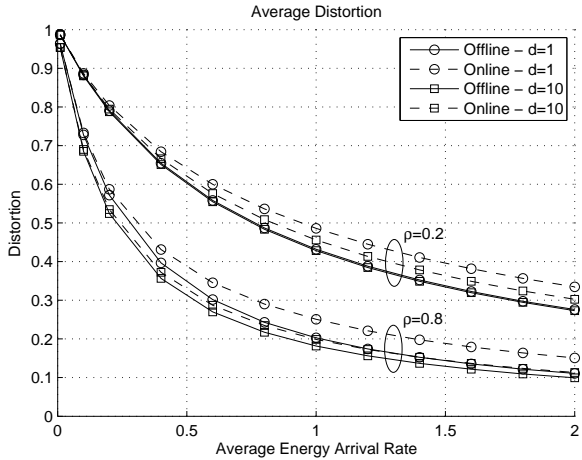}
    \caption{Average distortion for the offline and online policies, for low ($\rho=0.2$) and high ($\rho=0.8$) correlation, and delay-constrained ($d=1$) and delay-tolerant ($d=10$) scenarios.}
    \label{fig:online}
\end{figure}

\subsection{Convergence}

Next, we investigate the convergence properties of the proposed scheme. Specifically, in Figure \ref{fig:RelDistortion} we depict the relative error $\varepsilon$ between the average distortion at iteration $t$ and its optimal value, namely, $\varepsilon=\left|D_{avg}^{\star}-D_{avg}^{(t)}\right|/D_{avg}^{\star}$. For the update of the dual variables in \eqref{eq:update_dual}, we have used a time-varying step size\footnote{The step size used is $\alpha^{(t)}=\bar{\alpha}^{(t)}/\lVert g^{(t)}\rVert_{2}$, where $g^{(t)}$ is the corresponding subgradient and $\bar{\alpha}=1/\sqrt{t}$. This diminishing step size satisfies the convergence conditions given by $\bar{\alpha}^{(t)}\geq0$, $\lim_{t\to\infty}\bar{\alpha}^{(t)}=0$ and  $\sum_{t=1}^{\infty}\bar{\alpha}^{(t)}=\infty$ \cite{bertsekas1999nonlinear}.}. Clearly, convergence is slower for larger $d$ values. This stems from the fact that, for delay-tolerant scenarios, the search space for the solution is larger, as the summation in equation \eqref{eq:DistRateCap} evidences.

\begin{figure}[t]
    \centering
\includegraphics[width=\columnwidth]{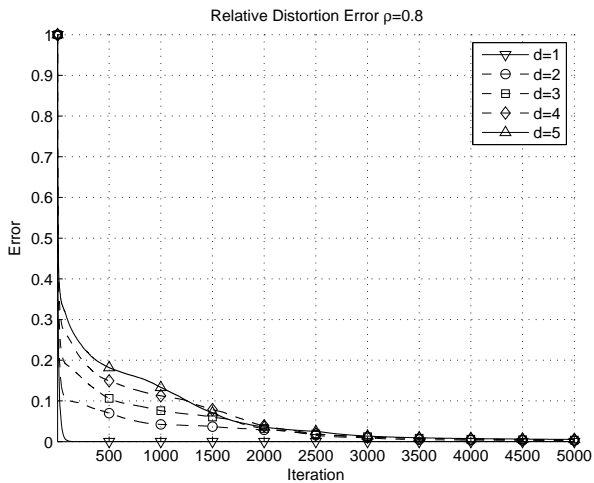} 
    \caption{Relative distortion error for $\rho=0.8$ and varying $d$.}
    \label{fig:RelDistortion}
\end{figure}

\subsection{Allocation of Individual Rates}

As discussed in Section \ref{sec:MinDist}, for $d>1$ there exists an infinite number of solutions for the allocation of the \emph{individual} rates (the system of equations is underdetermined). In order to get some insight on how individual rates are allocated, we will select the solution with the lowest 2-norm. This, clearly, penalizes solutions with very large (dissimilar) rates.

After solving the optimization problem \eqref{eq:OptProblem} and determining the optimal cumulative rates $r_{i,j}$, we find the individual rates $R_{i,j}$ by solving:
\begin{subequations}
\begin{align}
   \min_{\substack{\{R_{i,j}\} }}       \quad   &   \left(\sum_{i=1}^{K}\sum_{j=i-d+1}^{i}R^2_{i,j}\right)^{1/2}
                \\
   \mathrm{s.t.}        \quad   & r_{ii}=\sum_{j=i}^{i+d-1}R_{j,i},                     i=1,\ldots,K                        \\
                    & \sum\limits_{j=i-d+1}^{i}R_{i,j} \leq \log\left(1+|h_i|^2 p_i\right) ,    i=1,\ldots,K                        \\
                    & R_{i,j}\geq 0,                                    i=1,\ldots,K, j=i-d+1,\ldots,i
\end{align}
\label{eq:OptProblemIndivudalRates}
\end{subequations}
To that aim, we need to use $r_{ii}$ and $p_i$ from the solution of the (cumulative) rate and power allocation problem as an input (see first and second inequality constraints in the problem above).

Figure \ref{fig:individualRates} shows the allocation of \emph{individual} rates over the $d$ time slots for each source (a different color is used for each source). We consider scenarios with sources exhibiting low ($\rho=0.2$) and high ($\rho=0.8$) correlation. Interestingly enough, the higher the correlation, the lower the spread of individual rates over time slots (fewer sources in each time slot). This is consistent with the fact that, as discussed earlier, for low $\rho$ (and $d=1$) energy tends to be spent sooner. Accordingly, in delay-tolerant scenarios where the encoded data is transmitted in a number of time slots, when correlation is high the first time slots are favored.

\begin{figure}[t]
    \centering
    \includegraphics[width=\columnwidth]{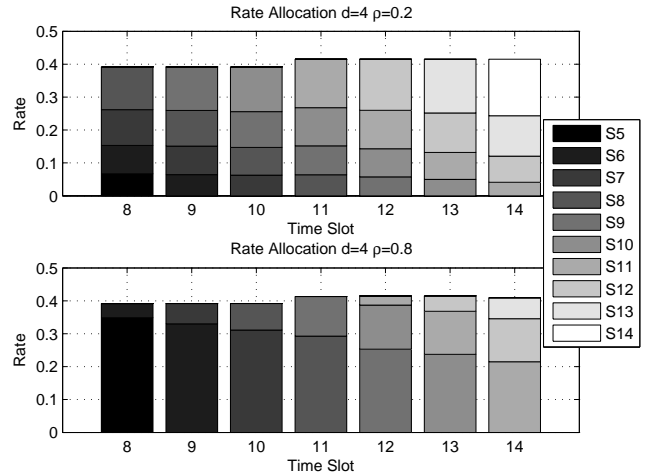}
    \caption{Allocation of individual rates for sources with low (top) and high (bottom) correlation ($K=20$; $d=4$; energy profile: $E_{1}=0.2$, $E_{2}=1$, $E_{4}=0.6$, $E_{6}=1$, $E_{7}=0.8$, $E_{8}=0.2$, $E_{9}=0.4$, $E_{11}=1.4$, $E_{13}=0.6$, $E_{14}=0.6$, $E_{16}=0.8$, $E_{17}=0.2$, $E_{18}=1$, $E_{19}=0.2$, and $E_{20}=0.4$).}
    \label{fig:individualRates}
\end{figure}

\section{Conclusions}
\label{sec:Conclusions}

In this paper, we have investigated the impact of source correlation in the design of point-to-point optimal transmission policies with energy-harvesting sensors. We have considered both delay-constrained delay-tolerant scenarios. In both cases, our goal was to minimize the average distortion in the decoded (reconstructed) observations by using data from previously encoded sources as side information. We have formulated the problems in a convex optimization framework. Besides, we have proposed an iterative procedure, based on the subgradient method, to solve both problems. Interestingly, the procedure entails the interaction of a directional and reverse water-filling schemes in each iteration. For the delay-constrained scenario, numerical results revealed that, differently from the uncorrelated case, minimizing the average distortion implies encoding observations as accurately as possible upon energy arrivals. This holds true even if the transmit power allocated to subsequent time slots is lower or, eventually, zero (and, thus, an increase in distortion in such time slots). For the delay-tolerant scenario, we have observed that as delay increases, the power allocation policy converges to the tightest string below the cumulative energy harvesting curve. And, also, that the average distortion decreases. In comparison with other schemes not exploiting correlated sources as side information, ours attains an average distortion which is substantially lower (with reductions of up to $25\%$ or $80\%$ for $d=1$ and $d=10$, respectively. We have also proposed a myopic online policy exhibiting a moderate performance gap (some $20\%$ for low correlation and delay-constrained scenarios) with respect to the offline (optimal) policy. Besides, we have found that the time needed for the algorithm to convergence is higher for delay-tolerant scenarios since the search space is substantially larger there. Finally, we have observed that for delay-tolerant scenarios, the higher the correlation, the lower the spread of individual rates over time slots.

\appendices

\section{Derivation of the Average Distortion in \eqref{eq:iDist}}
\label{app:Distortion}

For compactness, hereinafter we let $R_{i}\triangleq\sum_{j=i}^{i+d-1}R_{j,i}$ denote the rate assigned to the $i$-th source over its $d$ time slots; and $D_i = \sigma^2_{x_i|u_1,\ldots,u_i}$ the distortion for the $i$-th source which can be recursively expressed as \cite{kay1993fundamentals}:
\begin{align*}
\sigma^2_{x_k|u_{1},\ldots,u_k} &=\sigma^2_{x_k|u_1,\ldots,u_{k-1}}-\frac{\text{cov}^2(u_k,x_k|u_{1},\ldots,u_{k-1})}{\sigma^2_{u_k|u_{1},\ldots,u_{k-1}}}.
\end{align*}
We prove by induction that
\begin{align}
D_i     =\sigma^2_x
        \biggl(
        \left(1-\rho\right)\sum^{i}_{j=2}
        &\rho^{i-j}
        \mathrm{e}^{-\textstyle\sum_{k=j}^{i}\sum_{l=k}^{k+d-1}R_{l,k}}
        \biggr.\nonumber\\
         \biggl. +&\rho^{i-1}
        \mathrm{e}^{-\textstyle\sum_{k=1}^{i}\sum_{l=k}^{k+d-1}R_{l,k}}
        \biggr) \nonumber \\
        =\sigma^2_x
        \biggl(
        \left(1-\rho\right)\sum^{i}_{j=2}
        &\rho^{i-j}
        \mathrm{e}^{-\textstyle\sum_{k=j}^{i}R_{k}}
        \biggr.\nonumber\\
         \biggl. +&\rho^{i-1}
        \mathrm{e}^{-\textstyle\sum_{k=1}^{i}R_{k}}
        \biggr).
        \label{eq:distInd}
\end{align}
We start by showing this expression holds for the base case ($i=1$). That is
\begin{align*}
D_1 &
    = \sigma^2_{x_1|u_1} \\
    &= \sigma^2_x-\frac{\left(\sigma^2_x\right)^2}{\sigma^2_x+\sigma^2_{z_1}}
    = \sigma^2_x-\frac{\left(\sigma^2_x\right)^2}{\sigma^2_x+\dfrac{\sigma^2_x}{e^{R_1}-1}}
    = \sigma^2_x e^{-R_1},
\end{align*}
which satisfies expression \eqref{eq:distInd}. For the inductive step, assume expression \eqref{eq:distInd} is true for $i=n$. Then consider
\begin{align*}
D_{n+1} &=  \sigma^2_{x_{n+1}|u_{1},\ldots,u_{n+1}} \\
 &=\sigma^2_{x_{n+1}|u_1,\ldots,u_{n}}-\sigma^2_{x_{n+1}|u_1,\ldots,u_{n}}\left(1-e^{-R_{n+1}}\right) \\
 &= \sigma^2_{x_{n+1}|u_1,\ldots,u_{n}}e^{-R_{n+1}} \\
 &= \sigma^2_{\sqrt{\rho}x_n + w_n|u_1,\ldots,u_{n}}e^{-R_{n+1}} \\
 &= \left(\rho\sigma^2_{x_n|u_1,\ldots,u_{n}} +\sigma^2_{w_n|u_{1},\ldots,u_n}\right)e^{-R_{n+1}} \\
 &= \left(\rho D_n+\sigma^2_x\left(1-\rho\right)\right)e^{-R_{n+1}} \\
 &= D_n \rho e^{-R_{n+1}} +\sigma^2_x\left(1-\rho\right)e^{-R_{n+1}}
\end{align*}
Then by the induction hypothesis we have
\begin{align*}
D_{n+1}   =&\sigma^2_x
        \biggl(
        \left(1-\rho\right)\sum^{n}_{j=2}
        \rho^{n-j}
        \mathrm{e}^{-\textstyle\sum_{k=j}^{n}R_{k}}
        \biggr.\nonumber\\
         \biggl. +&\rho^{n-1}
        \mathrm{e}^{-\textstyle\sum_{k=1}^{n}R_{k}}
        \biggr)
        \rho e^{-R_{n+1}}+\sigma^2_x\left(1-\rho\right)e^{-R_{n+1}}
\end{align*}
and by rearranging terms we have
\begin{align*}
D_{n+1}    =\sigma^2_x
        \biggl(
        \left(1-\rho\right)\sum^{n+1}_{j=2}
        &\rho^{n+1-j}
        \mathrm{e}^{-\textstyle\sum_{k=j}^{n+1}R_{k}}
        \biggr.\nonumber\\
         \biggl. +&\rho^{n+1-1}
        \mathrm{e}^{-\textstyle\sum_{k=1}^{i+1}R_{k}}
        \biggr).
\end{align*}
Thus, expression \eqref{eq:distInd} holds for $i=n+1$. Therefore, by the principle of induction, expression \eqref{eq:distInd} holds for all $i$.

\section{Proof of Proposition \ref{prop:equivalence}.}
\label{app:Equivalence}

\begin{proof}
For notational convenience, let $c_i \triangleq\log(1+|h_i|^2p_i)$ denote the channel capacity in the $i$-th timeslot. First, we focus on the direct proof. Assuming that the LHS of \eqref{eq:EquivalentSystem}, which is given by the system of inequalities
\begin{subequations}
\begin{align}
	\textstyle\sum_{j=i}^{i+d-1}R_{j,i} &=r_{ii},		&  i=1,\ldots,K,&                   			\label{eq:LHS1}\\
	\textstyle\sum_{j=i-d+1}^{i}R_{i,j} &\leq c_i,    &  i=1,\ldots,K,&           					\label{eq:LHS2}\\
	R_{i,j}&\geq 0,                                       			&  i=1,\ldots,K,& j=1,\ldots,i-d+1,	\label{eq:LHS3}
\end{align}
\label{eq:EquivalentSystemLHS}
\end{subequations}
has a solution in terms of individual rates $R_{i,j}$, our goal is to find the system of inequalities in the RHS of \eqref{eq:EquivalentSystem}, namely
\begin{subequations}
\begin{align}
	\textstyle r_{ij} &\textstyle =\sum_{k=j}^{i}r_{kk},				&  i=1,\ldots,K, j=1,&\ldots,i-1,		\label{eq:RHS1}\\
	\textstyle r_{ij} &\textstyle \leq \sum_{k=j}^{i+d-1}c_k,		&  i=1,\ldots,K, j=1,&\ldots,i,          \label{eq:RHS2}\\
	r_{ij} &\geq 0,                                       				&  i=1,\ldots,K, j=1,&\ldots,i.			\label{eq:RHS3}
\end{align}
\label{eq:EquivalentSystemRHS}
\end{subequations}
The constraints \eqref{eq:RHS1} follow directly from the definition of the cumulative rates \eqref{eq:defCumulativeRates}. Constraint \eqref{eq:RHS3} is also straightforward since, from its definition in \eqref{eq:defCumulativeRatesPre}, the cumulative rates $r_{i,j}$ can be expressed as a summation of \emph{non-negative} (see \eqref{eq:LHS3}) individual rates $R_{i,j}$. As for \eqref{eq:RHS2}, note that from \eqref{eq:LHS2} each \emph{non-negative} individual rate can be upper-bounded as follows
\begin{align}
	R_{i,j}&\leq c_i,                                       			&  i=1,\ldots,K, j=1,\ldots,i-d+1.	\label{eq:IndividualCapacity}
\end{align}
Next, by direct substitution of the bounds \eqref{eq:IndividualCapacity} into the equalities \eqref{eq:LHS1}, we have that
\begin{align}
	\textstyle r_{ii}&\textstyle \leq\sum_{j=i}^{i+d-1}c_{j},		&  i=1,\ldots,K                   			\label{eq:IndividualCapacity2}
\end{align}
Finally, by substitution of \eqref{eq:IndividualCapacity2} into the definition of cumulative rates \eqref{eq:defCumulativeRates}, inequality \eqref{eq:RHS2} follows.

Consider now the converse. Assume that the RHS of \eqref{eq:EquivalentSystem}, which is also given by the system of inequalities \eqref{eq:EquivalentSystemRHS}, has a solution in terms of cumulative rates $r_{ij}$. Then, we want to prove that there exists a non-empty set of individual rates $R_{i,j}$ satisfying the inequalities \eqref{eq:EquivalentSystemLHS} (i.e., the RHS of \eqref{eq:EquivalentSystem} has a solution even if it might not be unique, as discussed earlier). To prove that, we focus on the more restrictive case where we force the capacity constraint \eqref{eq:LHS2} to be satisfied with equality. Hence, the first two constraints in \eqref{eq:EquivalentSystemRHS} become:
\begin{subequations}
\begin{align}
	\textstyle\sum_{j=i}^{i+d-1}R_{j,i} &=r_{ii},	&  i=1,\ldots,K                   	\label{eq:Equality1}\\
	\textstyle\sum_{j=i-d+1}^{i}R_{i,j} &=c_i,    &  i=1,\ldots,K.           			\label{eq:Equality2}
\end{align}
\label{eq:EqualitySystem}
\end{subequations}
The system of equations above can be rewritten in matrix form:
\begin{equation}
\mathbf{A}\mathbf{x}=\mathbf{b}
\label{eq:EqualitySystemMatrixForm}
\end{equation}
where we have defined the column vectors $\mathbf{x}\triangleq[R_{1,1},R_{2,1},\ldots,R_{K,K}]^T$ and $\mathbf{b}\triangleq[r_{1,1},\ldots,r_{K,K},c_{1},\ldots,c_{K}]^T$, and where matrix $\mathbf{A}$ is given by the $\{0,1\}$ entries yielding the summations in \eqref{eq:EqualitySystem}. Next, we resort to Farkas' lemma:
\begin{lemma}[Farkas' Lemma {[21]}]
	If $\mathbf{A} \in \mathbb{R}^{m \times n}$ and $\mathbf{b} \in \mathbb{R}^m$, then exactly one of the following holds:
	\begin{enumerate}[(i)]
		\item There exists $\mathbf{x} \in \mathbb{R}^n$ such that $\mathbf{A}\mathbf{x}=\mathbf{b}$ and $\mathbf{x} \geq 0$.
		\item There exists $\mathbf{y} \in \mathbb{R}^m$ such that $\mathbf{y}^T\mathbf{A} \geq 0$ and $\mathbf{y}^T\mathbf{b}<0$.
	\end{enumerate}
\end{lemma}
where the inequality $\mathbf{x} \geq 0$ is defined element-wise. Clearly, alternative (i) in the Farkas lemma states that, if it holds, a solution to the LHS in terms of individual rates exists. In the next paragraphs, we prove (by contradiction) that alternative (ii) does \emph{not} hold for our problem. To that aim, we start by defining $\mathbf{y}\triangleq[k_{r_{11}},\ldots,k_{r_{KK}},k_{c_{1}},\ldots,k_{c_{K}}]^T$. Assume that alternative (ii) holds. To satisfy the condition $\mathbf{y}^T\mathbf{A} \geq 0$, there must exist a nonnegative set of coefficients $k_{r_{ii}}$ and $k_{c_{i}}$ such that
\begin{align}
k_{r_{ii}}+k_{c_{j}} \geq 0,				\quad  i=1,\ldots,K, j=i,\ldots,i+d-1.					\label{eq:FarkasIneq2}
\end{align}
And, condition $\mathbf{y}^T\mathbf{b}<0$ can be rewritten as
\begin{align}
	\textstyle \sum_{i=1}^{K} k_{r_{ii}} r_{ii} + \sum_{i=1}^{K} k_{c_{i}} c_{i}  < 0.         \label{eq:FarkasIneq1}
\end{align}
Next, we will check that for any set of valid $k_{r_{ii}}$ and $k_{c_{j}}$ equation \eqref{eq:FarkasIneq1} does not hold. To that aim, we will determine the lowest possible value of the LHS of \eqref{eq:FarkasIneq1} subject to the inequalities given by \eqref{eq:FarkasIneq2}. In other words, we need to solve an optimization (minimization) problem with the LHS of \eqref{eq:FarkasIneq1} playing the role of the objective function and \eqref{eq:FarkasIneq2} as constraints. Since $r_{ii}$ and $c_i$ are nonnegative, this is a linear program. Hence, the solution will lie at the vertex of the feasible region defined by \eqref{eq:FarkasIneq2} \mbox{\cite[Chapter 7]{schrijver1998theory}}. Since the expressions \eqref{eq:FarkasIneq2} define a convex cone, its only vertex is given by
\begin{align}
k_{r_{ii}}+k_{c_{j}} = 0,				\quad  i=1,\ldots,K, j=i,\ldots,i+d-1.					\label{eq:FarkasIneq5}
\end{align}
By recursively analyzing the various equations in \eqref{eq:FarkasIneq5}, we conclude that necessarily
\begin{align}
k\triangleq k_{r_{ii}}=-k_{c_{i}},				\quad  i=1,\ldots,K.										\label{eq:FarkasIneq3}
\end{align}
That is, except for the sign, all the coefficients are identical. By replacing \eqref{eq:FarkasIneq3} into the LHS of \eqref{eq:FarkasIneq1}, the objective function in the optimization problem becomes
\begin{align}
	\textstyle k\sum_{i=1}^{K} r_{ii} - k \sum_{i=1}^{K} c_{i}.           									\label{eq:FarkasIneq4}
\end{align}
From \eqref{eq:RHS1} and \eqref{eq:RHS2}, we have that $\sum_{i=1}^K r_{ii} \leq \sum_{i=1}^K c_i$. That is, the sum of cumulative rates of all sources is \emph{below or equal to} the channel capacity over all time slots. However, since the objective function in \eqref{eq:OptProblem} is nonincreasing in all $r_{ij}$, the optimal solution of \eqref{eq:OptProblem} must satisfy $\sum_{i=1}^K r_{ii} = \sum_{i=1}^K c_i$ (i.e., with equality). This means that, necessarily, \eqref{eq:FarkasIneq4} is lower bounded by 0, hence, \eqref{eq:FarkasIneq1} does not hold and, in turn, alternative (ii) in the Farkas theorem does not hold either. This concludes the proof.
\end{proof}

\bibliographystyle{IEEEtran}
\bibliography{bib}

\end{document}